\documentclass[pra,floatfix,amsmath,twocolumn]{revtex4}
\usepackage{amssymb}
\usepackage{graphicx}
\usepackage{graphics}
\usepackage{amsmath}
\usepackage{amsthm}
\usepackage{color}

\def\x{\mathcal{X}}
\def\y{\mathcal{Y}}
\def\a{\mathcal{A}}
\def\b{\mathcal{B}}

\newcommand{\bea}{\begin{eqnarray}}
\newcommand{\eea}{\end{eqnarray}}

\def\bi{\begin{itemize}}
\def\ei{\end{itemize}}
\def\bc{\begin{center}}
\def\ec{\end{center}}

\def\C{\hbox{$\mit I$\kern-.7em$\mit C$}}
\def\R{\hbox{$\mit I$\kern-.6em$\mit R$}}

\newcommand{\one}{\mbox{$1 \hspace{-1.0mm}  {\bf l}$}}
\def\tr{\mathrm{tr}}

\def\ns{\mathcal{NS}}
\def\c{\mathcal{C}}
\def\q{\mathcal{Q}}
\def\l{\mathcal{L}}
\def\P{\textbf{P}}
\def\L{\textbf{L}}
\def\supp{\textrm{supp\,}}
\def\co{\textrm{Conv}}
\def\rank{\textrm{rank\,}}

\newtheorem{theorem}{Theorem}
\newtheorem{corollary}[theorem]{Corollary}
\newtheorem{lemma}[theorem]{Lemma}
\newtheorem{proposition}[theorem]{Proposition}

\newtheorem{observation}[theorem]{Observation}

\begin{document}

\title{
Shared randomness and device-independent dimension witnessing}
\author{Julio I. de Vicente} \email{jdvicent@math.uc3m.es}
\affiliation{Departamento de Matem\'aticas, Universidad Carlos III de
Madrid, Avda. de la Universidad 30, E-28911, Legan\'es (Madrid), Spain}

\begin{abstract}
It has been shown that the conditional probability distributions obtained by performing measurements on an uncharacterized physical system can be used to infer its underlying dimension in a device-independent way both in the classical and quantum setting. We analyze several aspects of the structure of the sets of probability distributions corresponding to a certain dimension taking into account whether shared randomness is available as a resource or not. We first consider the so-called prepare-and-measure scenario. We show that quantumness and shared randomness are not comparable resources. That is, on the one hand, there exist behaviours that require a quantum system of arbitrarily large dimension in order to be observed while they can be reproduced with a classical physical system of minimal dimension together with shared randomness. On the other hand, there exist behaviours which require exponentially larger dimensions classically than quantumly even if the former is supplemented with shared randomness. We also show that in the absence of shared randomness, the sets corresponding to a sufficiently small dimension are negligible (zero-measure and nowhere dense) both classically and quantumly. This is in sharp contrast to the situation in which this resource is available and explains the exceptional robustness of dimension witnesses in the setting in which devices can be taken to be uncorrelated. We finally consider the Bell scenario in the absence of shared randomness and prove some non-convexity and negligibility properties of these sets for sufficiently small dimensions. This shows again the enormous difference induced by the availability or not of this resource.
\end{abstract}

\maketitle

\section{Introduction}

Is it possible to estimate the degrees of freedom of an uncharacterized physical system? This question has received much attention in the last years in what is known as device-independent dimension witnessing (DIDW). It turns out that it is indeed possible to make tests about the underlying dimension of a physical system without making any assumption on it nor on the internal functioning of the measurement devices used to interact with it. Dimension estimates can be constructed based only on the measurement data, i.e.\ on the observed probabilities of obtaining certain outcomes conditioned on the different possible choices of measurement. These results are not only interesting from the fundamental point of view but also play a role in quantum information processing. Besides allowing for experimental tests of the physical dimension \cite{experiments}, which might be considered as a resource, these investigations allow to constrain the correlations that are achievable when the setting limits the underlying dimension of the physical systems used in a protocol. These scenarios are know as semi-device-independent quantum information processing: no assumption is made on the working of the devices nor on the physical systems used except for its dimension. Ideas from DIDW have allowed to prove the security of certain cryptographic schemes \cite{qkd} and to provide randomness-expansion protocols in this framework \cite{racs}. Moreover, DIDW is intimately related to the field of quantum communication complexity, which studies the minimal amount of communication parties have to exchange to successfully carry out distributed computational tasks \cite{review}. Indeed, communication can be quantified by the dimensionality of the physical systems used to encode the messages.

The first proposals for DIDW considered the Bell scenario of quantum nonlocality since violating Bell inequalities by a certain amount might require quantum systems of at least a certain dimension \cite{bellwit}. Subsequently, the structure of quantum correlations under dimensionality constraints has been extensively studied \cite{qdim}. Although other settings have been considered \cite{other}, a different general and simple formalism for DIDW was presented in \cite{gallego} in the so-called prepare-and-measure scenario, which has been largely explored afterwards \cite{dallarno,didwpm}. Both the Bell and the prepare-and-measure scenarios rely on different parties holding devices that interact with the physical system. It is usually assumed that the action of these devices might be correlated by the parties having access to a common random variable. This induces convexity into the sets of observable probability distributions corresponding to a given dimension and separation theorems can be used to obtain linear functionals that enable DIDW. However, shared randomness can be viewed as a resource and in certain settings it might be more natural to assume that all devices are independent (this is the case, for example, when the devices are trusted and are not jointly conspiring to mimic higher-dimensional behaviours). Conditions for DIDW with uncorrelated devices have been presented in \cite{brunner} (prepare-and-measure scenario) and more recently in \cite{sikora1} (Bell scenario) and \cite{sikora2} (prepare-and-measure scenario).

In this paper we explore the differences for DIDW considering whether shared randomness is available as a resource or not. In order to do this, we analyze in detail the structure of the sets of probability distributions corresponding to a certain dimension taking into account both possibilities. We first consider the prepare-and-measure scenario and we show that quantumness and shared randomness are not comparable resources. That is, on the one hand, there exist behaviours that require a quantum system of arbitrarily large dimension in order to be observed while they can be reproduced with a classical physical system of minimal dimension together with shared randomness. On the other hand, using results from communication complexity it can be seen that there exist behaviours which require exponentially larger dimensions classically than quantumly even if the former is supplemented with shared randomness. We also show another clear difference depending on whether shared randomness is available or not. In the absence of it, the sets corresponding to a sufficiently small dimension are negligible both classically and quantumly: they are zero-measure and nowhere-dense subsets in the set of all possible behaviours. However, this is never the case in the other setting as these sets are never negligible independently of how small the dimension might be. This negligibility property also explains the exceptional robustness of dimension witnesses when devices are taken to be independent as observed in \cite{brunner}. In the second part of this article, we consider the Bell scenario. The availability or not of shared randomness is known to make a difference and non-convexity results for the sets of observable probability distributions of a fixed dimension are known \cite{vertesi,wolfe}. Here we extend these results and prove systematically some non-convexity properties for these sets for sufficiently small underlying dimension when the parties do not have access to shared randomness. Furthermore, contrary again to the case of correlated devices, we also show that in this case these sets have measure zero and are nowhere dense in the set of all quantum behaviours. In order to obtain all these results we use some very simple dimension estimates based on the rank of a matrix.

\section{Prepare-and-measure scenarios}\label{pms}


The prepare-and-measure scenario \cite{gallego} for witnessing dimensions in a device-independent way is the following. There are two parties, Alice (or A) and Bob (or B), which receive respectively inputs $x$ and $y$ from finite alphabets $\x$ and $\y$. Their only chance to communicate is by A sending a classical or quantum physical system to B depending on her input. The dimension of this system, to be defined precisely below, quantifies the amount of communication used. Upon receival of the message, B interacts with the system by performing a measurement depending on his input and produces an output $b$ which can take values in a finite alphabet. For simplicity, we will consider this output to be binary, i.\ e.\ $b\in\{0,1\}$. Then, we can record the conditional probabilities with which each output occurs for any given pair of inputs: $P(b|xy)$. This is the main object in a device-independent scenario and we will refer to it as behaviour and denote it by $\P$. Of course, as conditional probabilities, behaviours are characterized by $P(b|xy)\geq0$ $\forall b,x,y$ and $\sum_bP(b|xy)=1$ $\forall x,y$.

The question to be addressed in this setting is the following. Without using any knowledge on how A and B process their information, what is the minimal amount of classical or quantum communication sent from A to B that is compatible with the observation of a given behaviour? The possible classical messages $m(x)$ are given by dits, i.\ e.\ $m\in\{1,\ldots,d\}$. Thus, the amount of classical communication is measured by the dimension of the message $d$. A always has the chance to use a random strategy, i.\ e.\ she can send a message $m$ given $x$ with probability $s(m|x)$, and so does B, i.\ e.\ he can produce an output $b$ given $y$ and the reception of $m$ with probability $t(b|ym)$. In the quantum case A sends quantum states $\rho_x$. The dimension of her message is thus
\begin{equation}\label{quantumdim}
d=\dim \sum_x\supp\rho_x,
\end{equation}
where $\supp$ stands for the support of an operator. In order to produce his output, B can interact with the message through a quantum measurement conditioned on his input.

Thus, we define the set of behaviours obtained by sending classical messages of dimension at most $d$ by $\c_d$ (this and the other sets to be defined below also depend on $|\x|$ and $|\y|$, which we drop to ease the notation since these quantities should be in general clear from the context). In other words, $\P\in\c_d$ when
\begin{equation}\label{setc}
P(b|xy)=\sum_{m=1}^{d}s(m|x)t(b|my).
\end{equation}
On the other hand, $\q_d$ denotes the set of behaviours achievable by sending quantum states of dimension at most $d$. That is, $\P\in\q_d$ if there exists measurements for B, $\{\Pi_b^y\geq0\}$ with $\sum_b\Pi_b^y=\one$ $\forall y$, such that
\begin{equation}\label{setq}
P(b|xy)=\tr(\rho_x\Pi_b^y)
\end{equation}
where the $\{\rho_x\}$ are of dimension less or equal to $d$ (cf.\ Eq.\ (\ref{quantumdim})).

As already mentioned in the introduction, this does not exhaust all possibilities. Depending on the physical setting A and B may be granted with another resource to build their strategy: shared randomness. This means that A and B may pre-establish the strategy each will follow depending on the value of a random variable they both have access to. This boils down to the fact that they can prepare any convex combination of their previously allowed behaviours. Thus, we define the sets of behaviours obtained by sending classical or quantum messages of dimension at most $d$ together with shared randomness by
\begin{equation}
\c'_d=\co(\c_d),\quad\q'_d=\co(\q_d),
\end{equation}
where $\co(\cdot)$ stands for the convex hull. The sets $\c_d$ and $\q_d$ need not be convex so in general we have strict inclusions $\c_d\subset\c'_d$ and $\q_d\subset\q'_d$ \cite{dallarno}. This means that shared randomness is indeed a resource which can allow to perform some tasks using less communication. On the other hand, we clearly have as well the inclusions $\c_d\subseteq\q_d$ and $\c'_d\subseteq\q'_d$, which can also be seen to be in general strict. That is, quantum strategies are also a resource over classical strategies in order to reduce the amount of communication.

\begin{figure}[t]
\includegraphics[scale=0.5]{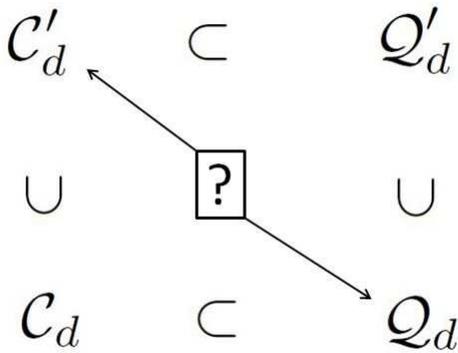}
\caption{For a fixed value of $d$, all inclusions among sets are clear except for those corresponding to classical communication together with shared randomness and quantum communication without shared randomness.}
\end{figure}

Notice that if $d\geq|\x|$, the scenario is trivial. A can unambiguously encode the value of her input into her message to B, who can then use his private randomness to output any possible behaviour. Hence, $\c_{|\x|}=\q_{|\x|}=\c'_{|\x|}=\q'_{|\x|}$, which constitute the set of all behaviours in a given setting. Thus, given any valid behaviour there always exist values of the dimension in which it is realizable in any of the aforementioned sets, since in the worst case we have $d=|\x|$. Therefore, it is always well defined to ask what the minimal value of the dimension is to obtain some behaviour in any of the four sets of possible strategies. Taking into account the inclusions pointed out above, it comes as a natural question what the relation between $\c'_d$ and $\q_d$ is (see Fig.\ 1). Moreover, given the status of both quantumness and shared randomness as a resource, it is interesting to know if one can exchange one for the other or if one is strictly more powerful than the other. Actually, this is a standard question in the context of communication complexity \cite{randvsquant}. Here, one usually seeks for differences in $d$ which are larger than a logarithmic cost over $|\x|$ as this is considered negligible with respect to the size of the input: the so-called exponential separations. We will show that there exist scenarios in which $\c'_d\nsubseteq\q_d$ and $\q_d\nsubseteq\c'_d$, with both separations being exponential (or even arbitrary). The second inclusion is a straightforward observation from known results in communication complexity. In order to establish the first one, we will first observe in the following subsection some very simple dimension estimates based on the rank of a matrix associated to $\P$. Using again these estimates, we will finish this section by showing the negligibility of low-dimensional sets in the absence of shared randomness. This explains the exceptional robustness to noise of dimension witnessing in this scenario and provides a clear contrast to the case where this resource is available.

\subsection{Dimension estimates}\label{estimates}

The fact that $\c'_{d}$ and $\q'_{d}$ are convex sets allows to separate each set from its complement by linear functionals on the behaviours. This gives rise to the so-called linear dimension witnesses \cite{gallego,dallarno}. The case of $\c_{d}$ and $\q_{d}$  was recently addressed in \cite{brunner}, which obtained some non-linear dimension witness for these non-convex sets. Specifically, they consider the scenario in which $|\x|=2|\y|=2k$ and show that the $k\times k$ matrix $W_k$ with entries $W_k(i,j)=P(0|2j-1,i)-P(0|2j,i)$ ($i,j=1,\ldots,k$) is such that $\det W_k=0$ for all behaviours in $\c_d$ ($\q_d$) with $d\leq k$ ($d\leq\sqrt{k}$). Thus, the determinant of $W_k$ being non-zero allows one to establish non-trivial lower bounds on the required dimensionality both in the classical and quantum case.

In the following we obtain more refined estimates. In order to do so and to deal with behaviours, we will arrange the array of numbers given by $\P$ into a matrix $P\in\mathbb{R}^{|\x|\times2|\y|}$ according to the rule
\begin{equation}\label{behavior}
P=\sum_{bxy}P(b|xy)|x\rangle\langle yb|,
\end{equation}
where in the standard notation of quantum mechanics $|yb\rangle=|y\rangle\otimes|b\rangle$ and $\{|y\rangle\}$ denotes the computational basis of $\mathbb{R}^{|\y|}$ and similarly for the other alphabet elements. In other words, $P$ takes the form
\begin{widetext}
\begin{equation}
P=\left(
    \begin{array}{ccccccc}
      P(0|11) & P(1|11) & P(0|12) & P(1|12) & \cdots & P(0|1|\y|) & P(1|1|\y|) \\
      P(0|21) & P(1|21) & P(0|22) & P(1|22) & \cdots & P(0|2|\y|) & P(1|2|\y|) \\
      \vdots & \vdots & \vdots & \vdots & \ddots & \vdots & \vdots \\
      P(0||\x|1) & P(1||\x|1) & P(0||\x|2) & P(1||\x|2) & \cdots & P(0||\x||\y|) & P(1||\x||\y|) \\
    \end{array}
  \right).
\end{equation}
\end{widetext}

Now, in the case $\P\in\c_d$, we can make the following immediate observation. Since the behaviour is given by Eq.\ (\ref{setc}), we have that
\begin{align}
P&=\sum_{m=1}^d\left(\sum_xs(m|x)|x\rangle\right)\left(\sum_{by}t(b|my)\langle yb|\right)\nonumber\\
&=\sum_{m=1}^du_mv_m^T
\end{align}
for some real (actually non-negative) vectors $\{u_m\}$ and $\{v_m\}$ of size $|\x|$ and $2|\y|$ respectively. Thus, we clearly see that if $\P\in\c_d$, then $\rank P\leq d$. On the other hand, if the behaviour is quantum, Eq.\ (\ref{setq}) tells us that its entries are given by the Hilbert-Schmidt inner product of pairs of Hermitian matrices of size $d\times d$. This set of matrices forms a subspace which is isomorphic to $\mathbb{R}^{d^2}$. Thus, there exists vectors $\{w_x\}$ and $\{t_{by}\}$ in this space such that $P(b|xy)=w_x^Tt_{by}$. Therefore, we now have that if $\P\in\q_d$, then $\rank P\leq d^2$.

\begin{observation}\label{obsest}
If $\P\in\c_d$, then $d\geq\rank P$ while if $\P\in\q_d$, then $d\geq\sqrt{\rank P}$.
\end{observation}

These simple observations generalize the previous aforementioned result of \cite{brunner} in several ways. First, we do not put any constraint on the size of the alphabets $\x$ and $\y$. Second, the estimates are finer since we do not rely on a matrix being of full rank or not but the bound is sensitive to the different possible values of the rank. It should be mentioned that the result of Bowles \emph{et al.} based on the $W$ matrix is also obtained by showing that the entries of this matrix are given by the inner product of a set of vectors. Thus, rank estimates are also possible in this case. In more detail, one obtains that $d\geq\rank W_k+1$ if $\P\in\c_d$ and $d\geq\sqrt{\rank W_k+1}$ if $\P\in\q_d$. Hence, it comes as a natural question whether it is better to use $P$ or $W$ to get the strongest estimate. Notice that, in the $|\x|=2|\y|=2k$ setting, the maximal possible rank of $P$ is $k+1$ (this is because several columns are surely linearly dependent due to the condition $\sum_bP(b|xy)=1$ $\forall x,y$) while for $W$ it is obviously $k$. Thus, in the case of maximal rank both approaches yield equal estimates. In the appendix we show that $\rank W\leq\rank P$. Thus, this suggests that it is generally better to use $P$. In fact, it can only be worse in cases for which $\rank W=\rank P$. However, this can only lead to a difference of one in the estimate (and not always in the quantum case since it holds for many natural numbers $n$ that $\lceil\sqrt{n+1}\rceil=\lceil\sqrt{n}\rceil$).

It is worth mentioning that stronger bounds on $d$ can be placed by using generalizations of the rank \cite{psd}. Recent literature has established an intimate relation between the non-negative rank and the classical dimension and the positive semidefinite rank and the quantum dimension in similar scenarios \cite{psd,ranks}. It is immediate to check that the nonnegative rank of $P$ and the positive semidefinite rank of $P$ are lower bounds for $d$ in the classical and quantum case respectively in the scenario considered here. Reference \cite{sikora2} also offers related strategies to bound the dimension. However, we stick here to the weaker rank estimates because they seem much easier to use. In fact, we do not know efficient algorithms to compute these other notions of the rank \cite{vavasis,psd}.

\subsection{Behaviours more expensive quantumly than classically together with shared randomness}\label{srmoreq}

In this subsection we will show that $\c'_d\nsubseteq\q_d$ with an arbitrary separation: there exist behaviours requiring a constant amount of classical communication together with shared randomness (actually just one bit) while the necessary quantum communication increases at least as $\sqrt{|\y|}$. In more detail, we will consider general scenarios such that $|\y|=k$ and $|\x|=m\geq k+1$ (this condition is only to make possible that the matrices of behaviours can have the largest possible rank, $k+1$) and we will construct behaviours $\P_k$ for any natural $k$ which they all belong to $\c'_2$ but cannot belong to $\q_{\lfloor\sqrt{k}\rfloor}$. The idea is to mix a sufficient number of behaviours in $\c_2$ such that the corresponding matrix $P_k$ has its rank as large as possible so that Observation \ref{obsest} leads us to conclude that $d\geq\sqrt{k+1}$ in order for $\P_k\in\q_d$ to hold.

An example of such a construction goes as follows. Here and throughout this paper we will denote by $e_i^{(n)}$ the vector of $\mathbb{R}^n$ that has zeroes everywhere except a 1 in the $i$th entry. Consider the behaviour $\textbf{D}_1$ in the aforementioned setting whose $m\times2k$ matrix is given by (to ease the notation we drop the dependence on $k$)
\begin{widetext}
\begin{equation}\label{ex}
D_1=\left(
      \begin{array}{cccc}
        (e_2^{(2)})^T & (e_1^{(2)})^T & \cdots & (e_1^{(2)})^T \\
        (e_1^{(2)})^T & (e_1^{(2)})^T & \cdots & (e_1^{(2)})^T \\
        \vdots & \vdots & \ddots & \vdots \\
        (e_1^{(2)})^T & (e_1^{(2)})^T & \cdots & (e_1^{(2)})^T \\
        (e_1^{(2)})^T & (e_1^{(2)})^T & \cdots & (e_1^{(2)})^T \\
        \vdots & \vdots & \vdots & \vdots \\
        (e_1^{(2)})^T & (e_1^{(2)})^T & \cdots & (e_1^{(2)})^T \\
      \end{array}
    \right)
    =\left(
        \begin{array}{c}
          0 \\
          1 \\
          \vdots \\
          1 \\
        \end{array}
      \right)\left(
               \begin{array}{cccc}
                 (e_1^{(2)})^T & \cdots & (e_1^{(2)})^T \\
               \end{array}
             \right)+e_1^{(m)}\left(
               \begin{array}{cccc}
                 (e_2^{(2)})^T & (e_1^{(2)})^T & \cdots & (e_1^{(2)})^T \\
               \end{array}
             \right).
\end{equation}
\end{widetext}
The second way to write $D_1$ shows clearly that $\textbf{D}_1\in\c_2$: B outputs all the time $b=0$ except maybe when $y=1$ depending on the bit sent by A, her action relying on whether she gets the input $x=1$ or any other. We can similarly define the behaviours $\textbf{D}_i$ ($i=1,\ldots,k$) whose matrices are all made by $1\times2$ blocks given by $(e_1^{(2)})^T$ except at the position $(i,i)$ where it is given by $(e_2^{(2)})^T$. By the same arguments as above, we have that $\textbf{D}_i\in\c_2$ $\forall i$. It is easy to see that any non-trivial mixture of all these behaviours has maximal rank. Taking for instance
\begin{equation}
\P_k=\sum_{i=1}^k\frac{1}{k}\textbf{D}_i,
\end{equation}
we find that
\begin{equation}
P_k=\left(
      \begin{array}{cccc}
        c_k^T & (e_1^{(2)})^T & \cdots & (e_1^{(2)})^T \\
        (e_1^{(2)})^T & c_k^T & \cdots & (e_1^{(2)})^T \\
        \vdots & \vdots & \ddots & \vdots \\
        (e_1^{(2)})^T & (e_1^{(2)})^T & \cdots & c_k^T \\
        (e_1^{(2)})^T & (e_1^{(2)})^T & \cdots & (e_1^{(2)})^T \\
        \vdots & \vdots & \vdots & \vdots \\
        (e_1^{(2)})^T & (e_1^{(2)})^T & \cdots & (e_1^{(2)})^T \\
      \end{array}
    \right),
\end{equation}
where
\begin{equation}
c_k^T=\left(
    \begin{array}{cc}
      1-1/k & 1/k \\
    \end{array}
  \right).
\end{equation}
Now, one can see that $P_k$ has the largest possible rank, i.\ e.\ $\rank P_k=k+1$. This is because, on the one hand, all even columns are clearly linearly independent, i.e.\ $col_{2j}(P_k)=e_j^{(m)}/k$ for $j=1,2,\ldots, k$. On the other hand, if we add to this set any other odd column, the set remains linearly independent because this column has non-zero entries where all the others have a zero entry (from the $(k+1)$th entry to the $m$th). Thus, using Observation \ref{obsest}, we finally obtain that if $\P_k\in\q_d$ it must hold that $d\geq\sqrt{k+1}$ while, by construction, $\P_k\in\c'_2$ $\forall k$. In passing, since obviously $\P_k\in\q'_2$ $\forall k$, this also shows the non-convexity of $\q_d$ (see also \cite{dallarno} and \cite{sikora2}).

\subsection{Behaviours cheaper quantumly than classically together with shared randomness}\label{qmoresr}

The results of \cite{brunner} discussed before show that there is a quadratic gap between $\c$ and $\q$. Interestingly, this gap can be seen to be exponential and extended to $\c'$. Testing classical and quantum dimensions in the prepare-and-measure scenario is intimately connected to the field of communication complexity when restricted to the scenario of one-way communication complexity. In fact, this setting is the same with the only difference that it is task-oriented. In this case, under the same restrictions A and B now have the goal of evaluating with high probability of success a binary function $f$ of their inputs that is known to both of them. That is, their strategies should aim at preparing behaviours $P(b|xy)$ for which the result $b=f(x,y)$ is much more likely than $b\neq f(x,y)$. The field of communication complexity studies what is the least amount of communication (from A to B in the one way case) necessary to evaluate different functions. The possible benefits of using quantum communication over classical communication have been extensively studied in the last years and there are several scenarios for which it is known that certain functions can be evaluated with a given probability of success requiring exponentially less communication in the quantum case than in the case of classical messages \cite{review}. Interestingly, the one-way scenario is not an exception and Refs.\ \cite{raz,montanaro} provide instances of this situation for the case of partial functions. In more detail, \cite{montanaro} considers a function, $f_P$, for which A receives an $n$-bit string $x$ (i.\ e.\ $|\x|=2^n$) and B a $n\times n$ permutation matrix $M$ (i.\ e.\ $|\y|=n!$). The goal is to output 1 if $Mx=x$ and 0 if $Mx$ and $x$ are sufficiently different (in a precise way which is irrelevant here). This is an example of a partial function or a function with a promise, A and B are guaranteed to receive a strict subset of the inputs $x$ and $y$ (those for which any of the above conditions hold). In \cite{montanaro}, it is shown that a quantum strategy solves this function using $O(\log n)$ qubits of communication (i.\ e.\ $d=O(n)$) while there cannot exist any classical strategy solving $f_P$ using less than of the order of $n^{7/16}$ bits. This immediately implies that $\q_d\nsubseteq\c'_d$. To see this, consider any behaviour corresponding to the aforementioned quantum strategies that solve $f_P$. It must then fulfill that $\P_n\in\q_d$ for some $d=O(n)$. However, it cannot be that $\P_n\in\c'_d$ as this would be in contradiction to the result of \cite{montanaro}. Indeed, any behaviour in $\c'_d$ cannot have the same entries as $\P_n$ over the subset of promised inputs $x$ and $y$ since it could then be used to solve $f_P$. Moreover, it must be that $\P_n\in\c'_{d'}$ with $d'$ scaling at least as $2^{n^{7/16}}$. Interestingly, there is roughly no difference between $\c_d$ and $\c'_d$ for the evaluation of functions. Newman's theorem \cite{newman} shows in the one-way communication scenario that classical strategies with shared randomness that solve some function can be turned into successful strategies without shared randomness with just a logarithmic overhead. Thus, if there exists some exponential gap for the solution of a function with quantum and classical resources, it must persist if we allow classical resources supplemented with shared randomness.

In the light of communication complexity, the reader might wonder whether the results of the previous section showing behaviours which required overwhelmingly more quantum communication to be prepared than classical communication together with shared randomness could be used to devise functions whose solution has a similar gap, i.\ e.\ functions that are at least exponentially cheaper to solve classically if shared randomness is allowed than quantumly. However, Newman's theorem forbids this possibility. In what comes to the evaluation of functions, the differences with and without shared randomness can be at most logarithmic in the one-way scenario even if one just uses classical messages.

\subsection{Structure of $\c_d$ and $\q_d$ and robustness of dimensionality detection}

As discussed in the introduction, the prepare-and-measure scenario was introduced to certify the dimension of uncharacterized physical systems in a device-independent way, i.e.\ based solely on the observed statistics and without any assumption on the internal working of the devices used. In this context any condition expressed in terms of the observed behaviour that guarantees that A and B exchange physical systems of at least a certain dimension is usually referred to as a dimension witness. The rank estimates introduced in Sec.\ \ref{estimates} are therefore an example of such an object. In practice, the measurement device cannot be perfectly isolated from external noise introducing errors in the experimentally reconstructed behaviour, which can make it less dimensional. The robustness of a dimension witness characterizes its noise tolerance in these scenarios and plays a crucial role in dimensionality certification. Although not strictly necessary, in this setting a natural assumption is that the preparing and measuring devices are uncorrelated (i.e.\ the preparer and the measurer are not maliciously conspiring to fool the certifier) and, hence, in this case one takes that shared randomness is not available. Thus, the problem here boils down to identifying what is the smallest $d$ such that $\P\in\c_d$ or $\P\in\q_d$. This was the motivation of \cite{brunner} to introduce the dimension witness based on the determinant of the matrix $W_k$ that we have reviewed in Sec.\ \ref{estimates}. It was observed there that these witnesses are extraordinarily robust tolerating arbitrary amounts of noise. In this section we investigate the structure of the sets $\c_d$ and $\q_d$ from this point of view and find reasons for this exceptional robustness. Low-dimensional sets are negligibly small in the set of all possible behaviours: they are nowhere dense and have measure zero. Hence, very contrived forms of noise are required to drastically reduce the dimension. This is in sharp contrast to the case where shared randomness is available since, as we also discuss here, the sets $\c'_d$ and $\q'_d$ are not negligible $\forall d\geq2$. We finish this section by observing that rank estimates are also extremely robust under any physically reasonable form of noise.

Notice that evaluating the rank of a matrix is an ill-conditioned problem. Due to the estimates presented in Sec.\ \ref{estimates}, this indicates that small perturbations of a behaviour could increase considerably the required dimension to prepare it. Furthermore, for general matrices it is well-known that lower-rank matrices are of measure zero and nowhere dense among matrices of higher rank. This suggests that lower-dimensional sets of behaviours might be negligible. We formalize this in the following.

\begin{theorem}\label{obsnegl}
In every scenario with $|\y|=k$ and $|\x|=m\geq k+1$, the sets $\c_d$ with $d<k+1$ and $\q_d$ with $d<\sqrt{k+1}$ have measure zero and are nowhere dense in the set of all possible behaviours $\c_m=\q_m$.
\end{theorem}
\begin{proof}
We will show that the set of rank-deficient behaviours (i.e.\ $\rank P<k+1$), which we will denote by $S$, is of measure zero and nowhere dense in $\c_m=\q_m$. Since in the classical case we have seen that $\rank P\leq d$, when $d<k+1$ we have that $\c_d\subset S$ and the result follows. The same applies to the quantum case. The proof of the claim for rank-deficient behaviours follows basically in a straightforward manner the analogous case for general matrices. 

Let us first show that $S$ has measure zero. Notice that the set of behaviours is a subset of $\mathbb{R}^{km}$ determined by specifying an arbitrary collection of values $0\leq P(0|xy)\leq1$ $\forall x,y$ and it has non-zero Lebesgue measure. When $P\in S$, this additionally imposes that the determinants of certain square submatrices vanish, which is a polynomial in the matrix entries, i.e.\ in the $\{P(0|xy)\}$. However, the zero set of a polynomial must have measure zero (unless it is the zero polynomial). Hence, $S$ has Lebesgue measure equal to zero.

Let us now see that $S$ is nowhere dense. For this we have to see that the closure of $S$, $\overline{S}$, has empty interior. Since we are dealing with finite-dimensional matrices, for these topological considerations we can take any matrix norm $||\cdot||$. First of all, it is useful to notice that $S$ is closed. This is because $S$ is characterized by the determinant of all $(k+1)\times(k+1)$ submatrices of $P$ being zero. Hence, the set is the preimage of a closed set under a continuous map (the determinant is a polynomial of the matrix entries) and it is therefore closed. Now, since $S=\overline{S}$, we just need to check that $S$ has an empty interior. Clearly, general rank-deficient matrices can always be approximated by full-rank matrices. It remains to see the same being careful that the full-rank approximation can be chosen to be a behaviour as well. For this, take any $P\in S$ and define for any $\epsilon\in[0,1]$
\begin{equation}\label{pepsilon}
P_\epsilon=(1-\epsilon)P+\epsilon Q,
\end{equation}
where
\begin{equation}\label{pepsilonq}
Q=\sum_{j=1}^{k}e_j^{(m)}v_j^T+\left(\sum_{j=k+1}^me_j^{(m)}\right)v_{k+1}^T
\end{equation}
with the $2k$-dimensional vectors $\{v_j\}$ defined by
\begin{widetext}
\begin{equation}\label{conspms}
v_1=\left(
      \begin{array}{c}
        e_2^{(2)} \\
        e_1^{(2)} \\
        \vdots \\
        e_1^{(2)} \\
      \end{array}
    \right), v_2=\left(
      \begin{array}{c}
        e_1^{(2)} \\
        e_2^{(2)} \\
        e_1^{(2)} \\
        \vdots \\
        e_1^{(2)} \\
      \end{array}
    \right),\ldots, v_k=\left(
      \begin{array}{c}
        e_1^{(2)} \\
        \vdots \\
        e_1^{(2)} \\
        e_2^{(2)} \\
      \end{array}
    \right), v_{k+1}=\left(
      \begin{array}{c}
        e_1^{(2)} \\
        \vdots \\
        e_1^{(2)} \\
        e_1^{(2)} \\
      \end{array}
    \right).
\end{equation}
\end{widetext}
Notice that $Q$ is a valid behaviour and, therefore, so is $P_\epsilon$ $\forall\epsilon$. It is important to notice that the $\{v_j\}$ are linearly independent (LI) vectors. That the first $k$ of them are LI is clear because each of them has a nonzero entry where all the others are zero. To see that adding $v_{k+1}$ to the set keeps it LI we notice the following. This vector has its second entry equal to zero, which is the case for all of the others except $v_1$. Thus if $v_{k+1}$ could be obtained as a linear combination of the other vectors, the weight of $v_1$ has to be zero. Iterating this argument for all even entries of $v_{k+1}$ we obtain the claim. Now, because $P$ and $Q$ are behaviours and the $\{v_j\}$ are nonnegative vectors, we have that $Pv_j$ and $Qv_j$ are nonnegative vectors too $\forall j$. Moreover, by construction $Qv_j\neq0$ $\forall j$ and, therefore, $P_\epsilon v_j\neq0$ $\forall j$ and all $\epsilon>0$. Since the $\{v_j\}$ are LI this implies that $\dim\ker P_\epsilon\leq k-1$ and, hence, given that the dimension of the kernel and the rank must add up to the number of columns $2k$, $\rank P_\epsilon=k+1$ $\forall\epsilon>0$, i.e.\ $P_\epsilon\notin S$ $\forall\epsilon>0$. Thus, we finally see that $\forall\delta>0$ and $\forall P\in S$, $\exists P'\notin S$ such that $||P-P'||<\delta$ (for this it suffices to take $P'=P_\epsilon$ with $\epsilon$ sufficiently small). Hence, $S$ does not contain any nonempty open set, i.e.\ it has empty interior as we wanted to prove \footnote{Actually, the fact that S is closed and has measure zero is already enough to prove that it is nowhere dense. However, similar constructions to that of Eq.\ (\ref{pepsilonq}) will be used through the paper.}.
\end{proof}

Thus, the sets $\c_d$ with $d<k+1$ and $\q_d$ with $d<\sqrt{k+1}$ are negligibly small and $\c_m\backslash\c_d$ and $\q_m\backslash\q_d$ have full measure and a dense interior. It might be that the conditions $d<k+1$ and $d<\sqrt{k+1}$ are an artifact of the proof due to the rank estimates and that the above claim can be extended to larger values of $d<m$. It could moreover be that $\c_{d-1}$ ($\q_{d-1}$) has zero measure and is nowhere dense in $\c_d$ ($\q_d$) for all $d$ such that $2\leq d\leq m$.

The result of Theorem \ref{obsnegl} is in sharp contrast to the case when shared randomness is available. The sets $\c'_d$ and $\q'_d$ are not negligible in the set of all possible behaviours $\forall d\geq2$, as we show below. Thus, this negligibility property provides a crucial difference for DIDW in the presence or not of shared randomness.

\begin{proposition} In every scenario with $|\y|=k$ and $|\x|=m\geq k+1$, the sets $\c'_d$ and $\q'_d$ $\forall d\geq2$ have nonzero measure and are not nowhere dense in the set of all possible behaviours $\c_m=\q_m$.
\end{proposition}
\begin{proof}
First of all, notice that $\c'_2\subset\c'_d$ ($d>2$) and $\c'_2\subset\q'_d$ ($d\geq2$). Hence, it suffices to prove the claim for $\c'_2$. Notice moreover that both $\c'_2$ and $\c_m$ are (convex) polytopes \cite{gallego}. We have discussed in the proof of Theorem \ref{obsnegl} that $\dim\c_m=km$. Therefore, one only needs to see that $\dim\c'_2=km$ too. For this, we have to find $km+1$ points in $\c'_2$ which are affinely independent. We give such a construction in the following. Notice that, as in Eq.\ (\ref{ex}), behaviours whose matrix has all except one rows equal are in $\c_2\subset\c'_2$. On the analogy of Eq.\ (\ref{ex}), we denote then by $\{\textbf{D}_{ij}\}$ ($1\leq i\leq m$, $1\leq j\leq k$) the behaviours whose $m\times2k$ matrices have $1\times2$ blocks equal to $(e_1^{(2)})^T$, except for the block at position $(i,j)$ which is equal to $(e_2^{(2)})^T$. We will also consider the behaviour $\textbf{D}_{0}$, whose matrix has all blocks equal to $(e_1^{(2)})^T$. Arguing in a similar manner as with the set of vectors of Eq.\ (\ref{conspms}), it is easy to see that the $km+1$ points $\{\textbf{D}_{0},\textbf{D}_{ij}\}$ (which all happen to be vertices of the polytope $\c'_2$) are LI and, hence, affinely independent.
\end{proof}

Theorem \ref{obsnegl} should not be interpreted as a physical impossibility of preparing low-dimensional behaviours. If the setting limits the amount of communication A can send to B, we are bound to observe such a low-dimensional behaviour. What we learn from it is that if the underlying dimension is sufficiently large, and B's measurements are subject to noise, it must have a very particular form in order to drastically reduce the dimension of the observed behaviour. Actually, for any behaviour $P$ one can see that under any physically reasonable form of noise $P_n$, the observed behaviour,
\begin{equation}\label{noisy}
P_\eta=\eta P+(1-\eta)P_n\quad(\eta\in[0,1]),
\end{equation}
maintains the rank $\forall\eta>0$. That is, $\rank P_\eta=\rank P$ $\forall\eta>0$ and the rank estimates are completely robust against noise. Thus, if $P$ can be certified to have dimension $d\geq k+1$ in the classical case (or $d\geq \sqrt{k+1}$ in the quantum case) by means of its rank, this will not change for its noisy version (unless in the extremal case of full noise $\eta=0$).

We finish this section by proving the above claim that $\rank P_\eta=\rank P$ $\forall\eta>0$. First we need to discuss what $P_n$ can be. The most reasonable and general form for the noise is that it is independent of A's input, i.e.\ $P_n(b|xy)=P_n(b|y)$ $\forall b,x,y$. This is because the errors only occur in the measurement process carried out by B and A has no control over it to affect the encoding of her message. Thus, $P_n={\bf 1}v^T$ where ${\bf 1}$ is a column vector in $\mathbb{R}^{|\x|}$ with all entries equal to 1 and $v=\sum_{b,y}P_n(b|y)|yb\rangle$. This implies that the noisy behaviour $P_\eta$ is given by a rank-one perturbation to $P$. This means that the rank of $P$ and $P_\eta$ can at most differ by one. However, we will see now that they are actually equal (as long as $\eta\neq0$). Since the image of $P_n$ is spanned by ${\bf 1}$, $\rank P_\eta=\rank P-1$ can only hold if there exists some vector $u$ such that $Pu\propto{\bf 1}$ in such a way that $P_\eta u=0$. Clearly, this vector must be of the form $u=\alpha{\bf 1}+w$ where $\alpha\in\mathbb{R}$ and $w\in\ker P$ because $P{\bf 1}=|\y|{\bf 1}$ for any behaviour. However, since we also have that $P_\eta{\bf 1}=|\y|{\bf 1}$, if it is possible to have a vector in the kernel of $P_\eta$ that is not in the kernel of $P$, $P_\eta u=[\alpha|\y|+(1-\eta)(v^Tw)]{\bf 1}=0$, it must be such that $v^Tw\neq0$. In this case we then have that $P_\eta w\neq0$, that is, we can also find a vector such that it is in the kernel of $P$ but not in that of $P_\eta$. This shows that $\rank P_\eta\neq\rank P-1$ when $\eta\neq0$. A similar argument shows that $\rank P\neq\rank P_\eta-1$ when $\eta\neq0$. Hence, we obtain the desired result.

\section{Bell scenarios}\label{bs}

The first scheme \cite{bellwit} that was proposed to test in a device-independent way the dimension of a quantum system used the Bell scenario of quantum non-locality \cite{reviewnl}. In this setting we also have two parties A and B, but they cannot communicate in this case. However, both A and B perform measurements dependent respectively on some inputs $x$ and $y$ on a bipartite quantum state $\rho$ they share. Each measurement leads to outputs $a$ and $b$ for A and B respectively. This scenarios can also be catalogued according to the (finite) size of the input and output alphabets $\x$, $\y$, $\a$ and $\b$. In the following we will consider that $|\x|=|\y|=m$ and $|\a|=|\b|=n$ and will refer to $(m,n)$ scenarios. Similarly to the prepare-and-measure setting, the object to which we have access to here is the set of conditional probabilities of obtaining the outputs $(a,b)$ given the choice of inputs $(x,y)$, $P(ab|xy)$. We will use again the term behaviour to refer to this collection of numbers. Obviously, it must hold that $P(ab|xy)\geq0$ $\forall a,b,x,y$ and $\sum_{a,b}P(ab|xy)=1$ $\forall x,y$. All behaviours attainable classically (together with shared randomness) satisfy
\begin{equation}
P(ab|xy)=\sum_\lambda p_\lambda P^A_\lambda(a|x)P^B_\lambda(b|y)\,\forall a,b,x,y,
\end{equation}
for some convex weights $\{p_\lambda\}$ and sets of conditional probabilities $\{P^A_\lambda\}$ and $\{P^B_\lambda\}$. Alternatively, the set of all such behaviours, the so-called local set $\l$, can be characterized to be the convex hull of all local deterministic behaviours (LDBs). The LDBs correspond to all possible deterministic uncorrelated behaviours, i.\ e.\ to those of the form $D(ab|xy)=\delta_{a,f(x)}\delta_{b,g(y)}$, where $f$ is any function mapping elements of $\x$ to $\a$ and similarly for $g$. That is, for every party a unique output occurs with probability 1 for every choice of input. Given a scenario, there is a finite number (actually $n^{2m}$) of possible LDBs and, hence, $\l$ is a polytope. We will denote by $\q$ here the set of all behaviours that can be obtained by performing measurements on bipartite quantum states $\rho_{AB}$, i.\ e.\
\begin{equation}
P(ab|xy)=\tr(\rho_{AB} E_a^x\otimes F_b^y)
\end{equation}
for some positive semidefinite operators $\{E_a^x,F_b^y\}$ such that $\sum_aE_a^x$ and $\sum_bF_b^y$ equal the identity in each party's Hilbert space $\forall x,y$. The celebrated conclusion of Bell's theorem is that $\l\subsetneq\q$.

For a fixed $(m,n)$ setting, one can now define $\q_d$ here as the set of all behaviours in $\q$ which are obtainable by a quantum state such that $\min_{X=A,B}(\dim\supp\rho_X)\leq d$, i.\ e.\ all behaviours obtainable by measuring quantum states of minimum local dimension at most $d$. As discussed in the introduction, the characterization of these sets has raised considerable attention from the point of view of both dimensionality certification and  semi-device-independent quantum information protocols. It is interesting to notice that, when the dimension of the physical system is not restricted, shared randomness is not a resource. Its availability is irrelevant in what comes to which behaviours can be observed with quantum preparations because $\q=\q_\infty$ is a convex set \cite{reviewnl}. However, this is not the case when the physical dimension of the underlying system plays a role. If shared randomness is freely available, this leads to consider the sets $\co(\q_d)$. Thus, it is interesting to explore the differences given by whether this resource is available or not and, in particular, whether $\q_d\neq\co(\q_d)$ in general. In fact, it is easy to see that $\q_1$ is non-convex in every scenario \cite{wolfe}. By definition, this set can only include uncorrelated behaviours. However, the convex hull of all LDBs gives rise to the full local polytope $\l$ and it is well-known that this set includes correlated behaviours. Reference \cite{vertesi} was the first one to observe that the sets $\q_d$ need not be convex in general. In more detail, it is shown therein that in the scenarios $(m,2)$ with $m$ even, every set $\q_d$ such that $d<\sqrt{m+1}$ is non-convex. In particular, this implies that $\q_2$ is non-convex already in the reasonably simple scenario $(4,2)$. It has been observed in \cite{wolfe} this to be the case even in the simplest possible scenario $(2,2)$ \footnote{The published version of \cite{wolfe} uses numerical evidence to provide a region in $\co(\q_2)$ which is not in $\q_2$. A posterior version [arXiv:1506.01119v4] mentions that the tools of \cite{sikora1} can be used to see that some points in this region are indeed not in $\q_2$.}. In the following we prove non-convexity properties of the sets $\q_d$ in the general scenario $(m,n)$. On the analogy of Sec.\ \ref{pms}, we will finish this section by proving that the sets $\q_d$ are negligible in $\q$ when the dimension is sufficiently small, a property which is not true for $\co(\q_d)$.

\subsection{Dimension estimates}

In order to prove the non-convexity and negligibility of $\q_d$ in $(m,n)$ scenarios we first derive dimension estimates. On the analogy of the previous section and adapting the techniques of \cite{vertesi}, we will obtain lower bounds on the quantum dimension in terms of the rank of some matrix associated to the behaviour. More explicitly, for every $(m,n)$ scenario we will arrange every behaviour $\P$ to form the $mn\times mn$ real matrix
\begin{equation}\label{Pmatrix}
P=\sum_{abxy}P(ab|xy)|xa\rangle\langle yb|,
\end{equation}
where in the standard notation of quantum mechanics $|xa\rangle=|x\rangle\otimes|a\rangle$ and $\{|x\rangle\}$ denotes the computational basis of $\mathbb{R}^{m}$ and similarly for the other alphabet elements. Thus, $P$ can be partitioned as a block matrix
$$P=\left(
      \begin{array}{ccc}
        P_{11} & \cdots & P_{1m} \\
        \vdots & \ddots & \vdots \\
        P_{m1} & \cdots & P_{mm} \\
      \end{array}
    \right)\in\mathbb{R}^{mn\times mn}$$
with blocks
$$P_{xy}=\left(
                           \begin{array}{ccc}
                             P(11|xy) & \cdots & P(1n|xy) \\
                             \vdots & \ddots & \vdots \\
                             P(n1|xy) & \cdots & P(nn|xy) \\
                           \end{array}
                         \right)\in\mathbb{R}^{n\times n}.$$
It will be relevant in the next subsection to note that the matrix associated to LDBs $D(ab|xy)=\delta_{a,f(x)}\delta_{b,g(y)}$ is of rank one, i.\ e.\
\begin{align}
D&=\left(\sum_{ax}\delta_{a,f(x)}|xa\rangle\right)\left(\sum_{by}\delta_{b,g(y)}\langle yb|\right)\nonumber\\
&=\left(\sum_{x}|xf(x)\rangle\right)\left(\sum_{y}\langle yg(y)|\right).
\end{align}

Suppose now that $\P\in\q_d$ and that the optimal quantum state is such that $d=\dim\supp\rho_A$. This means that the operators $\{E_a^x\}$ act on $\mathbb{C}^d$. Since they are Hermitian, they must then belong to a real vector space of dimension $d^2$ and, thus, at most $d^2$ of them can be linearly independent. In other words, we can express all the $\{E_a^x\}$ as real linear combinations of a fixed set of $d^2$ Hermitian operators (e.\ g.\ the identity and the generators of SU$(d)$). By linearity of the trace, this means that there are at most $d^2$ linearly independent rows in the matrix $P$ and, hence, $\rank P\leq d^2$. If it was the case that $d=\dim\supp\rho_B$, then we can make the same reasoning with the operators $\{F_b^y\}$ and the columns of $P$, arriving again at the same conclusion that $d\geq\sqrt{\rank P}$.

\begin{observation}\label{obsest2}
If $\P\in\q_d$, then $d\geq\sqrt{\rank P}$.
\end{observation}

It is important to notice for the following that the largest rank a matrix of a behaviour can attain is $mn-m+1$. This is because quantum behaviours must obey the no-signaling constraints
\begin{align}
\sum_bP(ab|xy)=\sum_bP(ab|xy'),\quad \forall a,x,y,y',\nonumber\\
\sum_aP(ab|xy)=\sum_aP(ab|x'y),\quad \forall b,x,x',y.
\end{align}
The set of all behaviours fulfilling these conditions will be denoted by $\ns$, which is also a polytope.

It should be mentioned that \cite{sikora1} already provides means to obtain lower bounds for the dimension and, actually, one can also use for these matters the positive semidefinite rank of $P$. However, as in the previous section, although weaker, rank estimates turn out to be more easily applied.

\subsection{Non-convexity of $\q_d$}

In order to prove the non-convexity of $\q_d$ we will use the following strategy. We will construct a local behaviour $\L$ such that $L$ has the largest possible rank $mn-m+1$. Through Observation \ref{obsest2} this implies that $\L\notin\q_d$ if $d<\sqrt{\rank L}$. However, since all local behaviours can be written as a convex combination of LDBs, it must hold that $\L\in\co(\q_1)\subset\co(\q_d)$ $\forall d$. Thus, for sufficiently small values of $d$, $\q_d$ cannot be convex.

\begin{lemma}\label{obsbell}
In every $(m,n)$ scenario there exists $\L\in\l$ such that $\rank L=mn-m+1$.
\end{lemma}
\begin{proof}
Consider the set of $n-1$ vectors of size $mn$ given by (we use here the same notation for the $\{e_i^{(n)}\}$ as in Sec.\ \ref{pms})
\begin{equation}
v_i^{(1)}=\left(
      \begin{array}{c}
        e_i^{(n)} \\
        e_1^{(n)} \\
        \vdots \\
        e_1^{(n)} \\
      \end{array}
    \right),\quad i=2,3,\ldots n.
\end{equation}
We will also consider similar sets of the same cardinality
\begin{equation}
\{v_i^{(2)}\}=\left\{\left(
      \begin{array}{c}
        e_1^{(n)} \\
        e_i^{(n)} \\
        e_1^{(n)} \\
        \vdots \\
        e_1^{(n)} \\
      \end{array}
    \right)\right\},\ldots, \{v_i^{(m)}\}=\left\{\left(
      \begin{array}{c}
        e_1^{(n)} \\
        \vdots \\
        e_1^{(n)} \\
        e_i^{(n)} \\
      \end{array}
    \right)\right\}.
\end{equation}
Notice now that the set $\{v_i^{(j)}\}$ ($i=2,\ldots n$, $j=1,\ldots m$) contains $mn-m$ LI vectors. This can be seen by noticing the each vector has a nonzero-entry where all the others are zero. Notice moreover that if we add the vector
\begin{equation}
v^{(0)}=\left(
      \begin{array}{c}
        e_1^{(n)} \\
        e_1^{(n)} \\
        \vdots \\
        e_1^{(n)} \\
      \end{array}
    \right)
\end{equation}
to this set, the vectors are still LI (cf.\ the reasoning after Eq.\ (\ref{conspms})). Finally, notice that the matrices
\begin{equation}
L_0=v^{(0)}(v^{(0)})^T,\quad \{L_{ij}\}=\{v_i^{(j)}(v_i^{(j)})^T\}
\end{equation}
clearly correspond to LDBs in the $(m,n)$ scenario. Hence, the matrix
\begin{equation}\label{lesp}
L=\frac{1}{mn-m+1}\left(L_0+\sum_{ij} L_{ij}\right)
\end{equation}
corresponds to a local behaviour and has the desired property that $\rank L=mn-m+1$. This is because by construction $Lv^{(0)}\neq0$ and $Lv_i^{(j)}\neq0$ $\forall i,j$ and, thus, $\dim\ker L\leq m-1$, which leads to the claim using that $\rank L+\dim\ker L=mn$. To see that indeed none of the above vectors is in the kernel of $L$, notice that $L_0v^{(0)}=mv^{(0)}$ while $L_{ij}v^{(0)}$ is a nonnegative vector $\forall i,j$ and similarly for the $\{v_i^{(j)}\}$.
\end{proof}
This shows that $\L\notin\q_d$ for $d<\sqrt{mn-m+1}$ and as discussed above we obtain the following corollary.
\begin{corollary}\label{cor}
Every set $\q_d$ in a $(m,n)$ scenario such that $d<\sqrt{mn-m+1}$ is not convex.
\end{corollary}

Notice that in the case $n=2$ we recover the result of \cite{vertesi} that allowed to verify the non-convexity of $\q_2$ in the scenario $(4,2)$. With our result, the simplest scenario for which we can see that $\q_2$ is not convex is $(2,3)$. However, as mentioned before, $\q_2$ is known to be non-convex in the simplest possible scenario $(2,2)$. Since the maximal rank of the matrix of a behaviour is $mn-m+1$, Corollary \ref{cor} cannot be further improved using our techniques. Thus, more powerful constraints could be in principle established going beyond the estimates based on the rank.

Lemma \ref{obsbell} has also a similar interpretation to the result of Sec.\ \ref{srmoreq} in the prepare-and-measure scenario. If A and B are bound to local preparations but have access to shared randomness, they can obtain behaviours which in order to be accessible quantumly without this resource need an arbitrarily large dimension as the number of inputs and/or outputs grows. An analogous result to that of Sec.\ \ref{qmoresr} is also obviously true by Bell's theorem. There are behaviours observable quantumly without shared randomness (even with the smallest possible dimension $d=2$) which cannot be attained by local strategies no matter how much access to this resource they have.

\subsection{Negligibility of low dimensional sets}

As in the prepare-and-measure scenario, one can show as well that a significant difference between the availability or not of shared randomness is that in the latter case the sets of low-dimension behaviours are negligible in the set of all quantum behaviours.

\begin{theorem}\label{obsnegl2}
Every set $\q_d$ in a $(m,n)$ scenario such that $d<\sqrt{mn-m+1}$ is of measure zero and nowhere dense in the full set of quantum behaviours $\q$.
\end{theorem}
\begin{proof}
The proof is given in two parts. We first show that with the given premise $\q_d$ is of measure zero and nowhere dense in $\ns$. We then show that this implies the claim.

The first part follows closely the proof of Theorem \ref{obsnegl} and we will only outline it. Similarly, we consider the set $S$ of rank-deficient no-signaling behaviours (i.e.\ $\rank P<mn-m+1$) and prove the claim for this set, which is extended to $\q_d$ with $d<\sqrt{mn-m+1}$ because $\q_d\subset S$. The no-signaling set $\ns$ is a polytope in $\mathbb{R}^t$ with \cite{pironio}
\begin{equation}\label{t}
t=m^2(n-1)^2+2m(n-1).
\end{equation}
For behaviours in $S$ some polynomials of the $t$ variables must additionally vanish and $S$ has measure zero in $\ns$. To see that $S$ is nowhere dense in $\ns$, one should follow the same argumentation as before replacing $P_\epsilon$ in Eq.\ (\ref{pepsilon}) by
\begin{equation}
P_\epsilon=(1-\epsilon)P+\epsilon L,
\end{equation}
where $L$ is given by Eq.\ (\ref{lesp}).

We finish by showing that $\q_d$ ($d<\sqrt{mn-m+1}$) having zero measure and being nowhere dense in $\ns$ implies the same negligibility properties inside $\q$. Since the local polytope $\l$ has the same dimension ($t$) as $\ns$ \cite{pironio} and $\l\subset\q\subset\ns$, we have that $\q$ is not of measure zero in $\ns$ \footnote{The fact that $\q$ is a bounded convex set guarantees that it is measurable.}. Hence $\q_d$ has measure zero in $\q$ as well. Regarding nowhere density, we use again that $\q$ is full dimensional together with the fact that it is a convex set. Corollary 6.4.1 in \cite{rockafellar} tells us then that for every $P$ in the interior of $\q$, $\mathop{\q}\limits^\circ$, we have that
\begin{equation}\label{ast}
\exists\epsilon>0\textrm{ such that }||P-P'||<\epsilon\Longrightarrow P'\in\q.
\end{equation}
Let us proceed by contradiction and assume that $\q_d$ is not nowhere dense in $\q$. Then, there would exist a $P\in\overline{\q}_d$ and $\delta>0$ such that $||P-P'||<\delta$ and $P'\in\q$ implies that $P'\in\overline{\q}_d$. By definition, such $P$ must belong to $\mathop{\overline{\q}_d}\limits^\circ$ and since $\q_d\subset\q$, it holds then that $\mathop{\overline{\q}_d}\limits^\circ\subset\mathop{\overline{\q}}\limits^\circ=\mathop{\q}\limits^\circ$, where the equality follows from the convexity of $\q$ (cf.\ Theorem 6.3 in \cite{rockafellar}). Thus, by condition (\ref{ast}) we can drop the assumption $P'\in\q$ if we take $\min\{\epsilon,\delta\}$, i.e.\
\begin{equation}
||P-P'||<\min\{\epsilon,\delta\}\Longrightarrow P'\in\overline{\q}_d.
\end{equation}
This means that $\q_d$ is not nowhere dense in $\ns$ and we have reached a contradiction \footnote{An alternative way to prove that $\q_{d}$ is nowhere dense in $\q$ would be to show that $\q_d$ is closed.}.
\end{proof}

Theorem \ref{obsnegl2} tells us that in such simple scenarios as $(4,2)$ or $(2,3)$ it is not only not enough to consider quantum systems with $d=2$ to reproduce all quantum behaviours but that this is almost never the case.

This negligibility property is again in sharp contrast to the case in which the devices of the parties can be correlated. When shared randomness is granted to the parties we have that $\l\subset\co(\q_d)$ $\forall d$. Since, as we have already used in the proof of Theorem \ref{obsnegl2}, $\dim \l=\dim \q=\dim \ns$, the sets $\co(\q_d)$ have non-zero measure and are not nowhere dense in $\q$ $\forall d$.

Looking at Theorem \ref{obsnegl2}, it comes as natural question whether the bound $d<\sqrt{mn-m+1}$ is optimal for the negligibility property to hold. Although we cannot answer completely this question, the above observation allows to establish a lower bound on $d$ for which negligibility does not hold anymore in every scenario \footnote{Notice that a similar reasoning can be applied to the prepare-and-measure scenario. However, it yields a trivial bound since in this case $d=|\x|$ is enough to generate the set of all behaviours.}.

\begin{proposition}
Every set $\q_d$ in a $(m,n)$ scenario such that $d\geq m^2(n-1)^2+2m(n-1)+1$ has non-zero measure and is not nowhere dense in the full set of quantum behaviours $\q$.
\end{proposition}
\begin{proof}
We first notice that if $\P\in\q_d$ and $\P'\in\q_{d'}$, then $\P_\lambda=\lambda\P+(1-\lambda)\P'\in\q_{d+d'}$ $\forall\lambda\in(0,1)$. This is a well-known argument that is used to show that $\q$ is convex. Indeed, if
\begin{align}
P(ab|xy)&=\tr(\rho E_a^x\otimes F_b^y),\nonumber\\
P'(ab|xy)&=\tr(\rho' (E')_a^x\otimes (F')_b^y),
\end{align}
then
\begin{equation}
P_\lambda(ab|xy)=\tr[(\lambda\rho\oplus(1-\lambda)\rho') \mathcal{E}_a^x\otimes \mathcal{F}_b^y]
\end{equation}
with $\mathcal{E}_a^x=E_a^x\oplus(E')_a^x$ $\forall a,x$ and $\mathcal{F}_b^y=F_b^y\oplus(F')_b^y$ $\forall b,y$.

On the other hand, Carath\'eodory's theorem \cite{rockafellar} tells us that $\forall \P\in\l=\co(\q_1)$ we have the convex combination
\begin{equation}
P=\sum_{i=1}^{t+1}\lambda_iP_i,\quad P_i\in\q_1\,\forall i,
\end{equation}
where we are using that $\dim\l=t$ (cf.\ Eq.\ (\ref{t})). This means that $\co(\q_1)\subset\q_{t+1}$ and since $\co(\q_1)$ is not of measure zero nor nowhere dense, we obtain the claim.
\end{proof}

As in the prepare-and-measure scenario, Theorem \ref{obsnegl2} also implies an exceptional robustness in the presence of noise for DIDW when shared randomness is not available. Notice that in this case it is natural to assume that the noisy behaviour $P_\eta$ (cf.\ Eq.\ (\ref{noisy})) is subjected to uncorrelated noise, i.e.\ $P_n(ab|xy)=P_A(a|x)P_B(b|y)$ $\forall a,b,x,y$, since this affects independently the devices held by A and B. Therefore, the noise induces again a rank-one perturbation and we have that $\rank P -1\leq\rank P_\eta\leq\rank P+1$ $\forall\eta>0$.

\section{Conclusions}

The task of DIDW has been receiving a lot of attention in recent years. It enables experiments to certify the underlying dimension of an uncharacterized physical system and it provides a framework for semi-device-independent quantum information processing. The most common scenarios for DIDW involve different parties interacting with the physical system: the so-called prepare-and-measure and Bell scenarios. Depending on the context it might or might not be the case that the parties are provided with an extra resource, shared randomness, that allows to correlate the different devices the parties hold. In this work we have explored the differences that may arise for the task of DIDW in these two possible settings. We have seen that shared randomness is indeed a powerful resource: certain behaviours which can be obtained by sending just one classical bit (when the devices are correlated) need quantum systems of arbitrarily large dimension in the absence of shared randomness (the necessary quantum dimension grows with the number of possible inputs while the classical dimension remains 2). On the other hand, quantumness is also more powerful than classical systems even if the latter have access to shared randomness. There are behaviours that require exponentially larger classical dimension even though in the quantum setting the devices are not correlated. We have also shown that one of the main differences given by the availability or not of this resource is not only the lack of convexity of the corresponding sets of probability distributions but the fact that for sufficiently small dimensions these sets are negligibly small (of measure zero and nowhere dense in the set of all possible distributions) if shared randomness is not granted. These results are obtained using very simple estimates for the dimension based on the rank of a matrix. For the future, it would be interesting to study whether the bounds on the dimension as a function of the number of possible inputs provided here for the sets to be non-convex or negligible in both the prepare-and-measure and Bell scenarios can be improved by using more sophisticated tools. The results of \cite{sikora1,sikora2} and the notions of non-negative and positive semidefinite rank might be helpful in this task.

\begin{acknowledgments}
I thank A. Monr\`{a}s, C. Palazuelos and I. Villanueva for very useful discussions. This research was funded by the Spanish MINECO through grants MTM2014-54692-P and MTM2014-54240-P and by the Comunidad de Madrid through grant QUITEMAD+CM S2013/ICE-2801.
\end{acknowledgments}

\begin{appendix}

\section{Relation between the ranks of $W$ and $P$}

Taking the matrices defined in Sec.\ \ref{estimates}, here we prove that $\rank W\leq\rank P$. In order to transform $P$ into $W$ we have to subtract to every odd row $i$ its subsequent row $i+1$. This a so-called elementary row operation and can be achieved by multiplying $P$ from the left with the matrix $E_i$, which is like the identity with the difference that the entry $(i,i+1)$ should be equal to $-1$. Since the matrices $\{E_i\}$ are all full-rank, the matrix $\prod_{i\textrm{ odd}}E_iP$ has the same rank as $P$. To obtain $W$ it just remains to delete all even rows and columns, a process which certainly cannot increase the rank. Thus, we obtain the desired result.

\end{appendix}

\end{document}